\documentclass[a4paper,11pt]{llncs}
\usepackage{latexsym,amssymb,amsmath,amsfonts,ifthen,owna4}

\newenvironment{asm}%
{\begin{quote}\begin{tabbing}1xxxx \= 2xxxx \= 3xxxx \= 4xxxx \= 5xx \= 6xx \= 7xx \= 8xx \= 9xx \kill}%
{\end{tabbing}\end{quote}}

\newcommand{\as}{\textit{asm\/}}

\title{Concurrent Computing with Shared Replicated Memory\thanks{The research reported in this article was partly supported by the {\bf Austrian Science Fund} for the project \emph{Behavioural Theory and Logics for Distributed Adaptive Systems} (\textbf{FWF: [P26452-N15]}).}}

\subtitle{Rigorous Specification and Analysis Using Concurrent Communicating Abstract State Machines}

\author{Klaus-Dieter Schewe\inst{1}, Andreas Prinz\inst{2}, Egon B\"orger\inst{3}}

\institute{
Zhejiang University, UIUC Institute, China,
\email{kdschewe@acm.org}
\and University of Agder, Department of ICT,
Agder, Norway, \email{andreas.prinz@uia.no}
\and Universit\`{a} di
Pisa, Dipartimento di Informatica, Pisa, Italy,
\email{boerger@di.unipi.it}}

\begin{document}

\maketitle

\begin{abstract}

The behavioural theory of concurrent systems states that any concurrent system can be captured by a behaviourally equivalent concurrent Abstract State Machine (cASM). While the theory in general assumes shared locations, it remains valid, if different agents can only interact via messages, i.e. sharing is restricted to mailboxes. There may even be a strict separation between memory managing agents and other agents that can only access the shared memory by sending query and update requests to the memory agents. This article is dedicated to an investigation of replicated data that is maintained by a memory management subsystem, whereas the replication neither appears in the requests nor in the corresponding answers. We show how the behaviour of a concurrent system with such a memory management can be specified using concurrent communicating ASMs. We provide several refinements of a high-level ground model addressing different replication policies and internal messaging between data centres. For all these refinements we analyse their effects on the runs such that decisions
concerning the degree of consistency can be consciously made.

\end{abstract}

\section{Introduction}

Abstract State Machines (ASMs) have been used since their discovery in the 1990s to model sequential, parallel and concurrent systems (see \cite[Ch.~6,~9]{boerger:2003}). For sequential systems the celebrated sequential ASM thesis \cite{gurevich:tocl2000} provides an elegant theoretical underpinning showing that every sequential algorithm as stipulated by three simple, intuitive postulates are captured by sequential ASMs. This was generalised in \cite{ferrarotti:tcs2016} for (synchronous) parallel systems\footnote{The behavioural theory proven by Ferrarotti et al. simplifies the previously developed parallel ASM thesis \cite{blass:tocl2003,blass:tocl2008} by exploiting the idea of multiset comprehension terms in bounded exploration witnesses, which was stimulated by previous research on non-deterministic database transformations \cite{schewe:ac2010,schewe:jucs2010}.}, and in \cite{boerger:ai2016} for asynchronous concurrent systems\footnote{This closed the gap in the behavioural theory of concurrent systems, as the definition of partially ordered runs in \cite{gurevich:lipari1995} was generally considered to be insufficient. There are examples of concurrent systems satisfying the intuitively clear property of sequential consistency \cite{lamport:1990} without having partially ordered runs.}, in which a concurrent system is defined by a family of algorithms assigned to agents that is subject to a concurrency postulate\footnote{Though the proof of the concurrent ASM thesis was first only conducted for families of sequential algorithms, the generalisation to families of parallel algorithms does not cause serious difficulties as sketched in \cite{schewe:acsw2017}.}.

This characterisation can be applied to many different models of concurrent computation (see e.g. \cite{agha:1986,best:1996,genrich:tcs1981,mazurkiewicz:lncs1987,tanenbaum:2007,winskel:1995}). While the thesis has been grounded on the assumption that shared locations are used, it was shown in \cite{boerger:jucs2017} that the theory remains valid, if different agents can only interact via messages, i.e. sharing is restricted to mailboxes. This includes the case of a strict separation between memory managing agents and other agents that can only access the shared memory by sending query and update requests to the memory agents. Naturally, the expectation of an agent $a$ sending a request to a memory management agent $b$ is that the result is the same as if there had been a concurrent run with agent $a$ executing its request directly on the shared locations.

However, as observed in \cite{prinz:abz2014} this expectation can be violated by relaxed shared memory management, in particular in combination with conflicting updates or data replication. Regarding conflicting updates, if agents $a$ and $a^\prime$ try simultaneously to update the value at a location $\ell$ to some new values $v$ and $v^\prime$, respectively, with $v \neq
v^\prime$, a concurrent run would be discarded. Instead of this, a memory management agent $b$ receiving update requests from $a$ and $a^\prime$ might randomly choose one of the possible values $v$ or $v^\prime$, so for one of the agents there would be a lost update and its latest step would not have been executed properly. 

Therefore, in \cite{boerger:ai2016} it is requested that ``if the effects of such techniques are behaviourally relevant for the problem \dots, then this should be described by a precise model through which a programmer can understand, check and justify that the code does what it is supposed
to do''. For the random selection such a precise model can be easily defined by a memory agent $m$ with the following rule\footnote{Here a shared (set-valued) function symbol $\Delta$ denoting an update set is used.}:

\begin{asm}
\> \texttt{FORALL} $\ell$ \texttt{WITH} $\exists v . (\ell,v) \in \Delta $
\texttt{CHOOSE} $v \in \{ v^\prime \mid (\ell,v^\prime) \in \Delta \}$ \texttt{IN} 
           $\ell := v$ 
\end{asm}

If this behaviour is known, one can specify what an agent should do in case of possibly lost updates, e.g. use strict transactions to avoid the situation to arise\footnote{See e.g. \cite{boerger:scp2016} for an ASM specification of concurrent systems with transactional behaviour.}, or integrate algorithms for mutual exclusion \cite{an:mth2016,lynch:1996}, or ensure that an agent only continues after the memory agent $m$ (considered as part of the environment \cite{boerger:2003}) has terminated its step, etc.

In this article we investigate the case of replicated data maintained by a memory management subsystem, where the replication should not appear in the requests nor in the corresponding answers, which is a standard requirement in distributed databases with replication \cite[Chap.13, pp.459ff.]{oezsu:2011}. Consider first an example of undesirable behaviour with four agents $a_1,\ldots,a_4$ having the following respective programs (where ; denotes sequential execution, and $Print(x)$ means to read the value $x$ and to copy it to some output):
\[ x:=1 \mid y:=1 \mid Print(x);Print(y) \mid Print(y); Print(x) \]

Then there is no concurrent run where (1) initially $x=y=0$, (2) each agent makes once each possible move, (3) $a_3$ prints $x=1,y=0$, and (4) $a_4$ prints $x=0,y=1$. However, if $x$ and $y$ are replicated, say that there are always two copies, and an update by the programs $a_1$ or $a_2$ affects only a single copy, such a behaviour will indeed be enabled.

Therefore, our objective is that such behaviour must be understandable from the specification so that the developer of the system can see, whether the consequences of such behaviour can be tolerated or 
additional consistency assurance measures have to be included. We assume a rather simple model, where shared data is logically organised in relations with primary keys, and data can only be accessed by means of the primary key values. We further assume relations to be horizontally fragmented according to values of a hash function on the primary key values, and these fragments are replicated. Replicas are assigned to different nodes, and several nodes together form a data centre, i.e. that are handled by one dedicated data management agent\footnote{This is similar to the data organisation in the noSQL database system Cassandra \cite{rabl:pvldb2012}, but otherwise the Cassandra system is of no importance for this article. We may think of nodes as storage units and of data centres as physical machines managing them.}. 

In addition, values in replicas carry logical timestamps set by the data centres and stored in the records in the nodes of the memory management subsystem. We further adopt Lamport's simple approach for the maintenance of timestamps \cite{lamport:cacm1978}, which basically advances a clock, if its data centre receives a record with a future timestamp. This allows us to formulate and investigate policies that guarantee certain to-be-specified levels of consistency (as known from replication in relational databases). 

For retrieval of a set of records a specified number of replicas has to be read, and for each record always the one with the latest timestamp will be returned. Depending on how many replicas are accessed the returned records may be (in the strict sense) outdated or not. Likewise, for the update of a set of records timestamps will be created, and a specified number of replicas of the records will be stored. Success of retrieval or update will be returned according to specified read- and write-policies.

In Section \ref{sec:ground} we will first specify the behaviour of a concurrent system with shared data requiring that all agents interact with this subsystem for data retrieval and updates using appropriate \texttt{SEND} and \texttt{RECEIVE} actions. The memory management subsystem will be specified
by a separate collection of agents. In Section \ref{sec:consistency} we investigate a refinement concerning policies how many replicas are to be read or updated, respectively. We show that some combinations of replication policies enable {\em view compatibility}, i.e. data consistency, which formalises the expectation above. In Section \ref{sec:data} we refine our specification taking the communication between data centres into account, and address the enforcement of the read and write policies. We obtain a complete, though not necessarily correct refinement, and as a consequence view compatibility cannot be guaranteed anymore. In fact, we even show that view compatibility implies view serialisability. That is, without exploiting the possibility of transactions---at least for single read or write requests---consistency cannot be preserved. Finally, we conclude with a brief summary and outlook.

\section{Ground Model for Shared Memory Management with Replication}\label{sec:ground}

We assume some familiarity with Abstract State Machines (ASMs), which can be understood as a form of pseudo-code with well-founded semantics\footnote{Here we do not repeat the formal definition of the semantics of ASMs---detailed definitions can be found in the textbook~\cite[Sect.2.2/4]{boerger:2003}.}. The {\em signature} $\Sigma$ of an ASM is a finite set of function symbols $f$, each associated with an arity $ar_f$. A {\em state} $S$ is a set of functions $f_S$ of arity $ar_f$ over some fixed base set $B$, given by interpretations of the corresponding function symbol $f$. Each pair $(f,(v_1,\dots,v_{ar_f}))$ comprising a function symbol and arguments $v_i \in B$ is called a {\em location}, and each pair $(\ell,v)$ of a location $\ell$ and a value $v \in B$ is called an {\em update}. A set of updates is called an {\em update set}. The evaluation of terms is defined as usual by 
\[val_S(f(t_1,\dots,t_{ar_f})) = f_S(val_S(t_1),\dots,val_S(t_{ar_f})) . \]

\noindent
We say that $f_S(val_S(t_1),\dots,val_S(t_{ar_f}))$ is the value at location $(f,(val_S(t_1),\dots,val_S(t_{ar_f})))$ in state $S$. ASM {\em rules} $r$ are composed using 

\begin{quote}\begin{description}

\item[\bf assignments.] $f(t_1,\dots,t_{ar_f}) := t_0$ (with terms $t_i$ built over $\Sigma$),

\item[\bf branching.] \texttt{IF} $\varphi$ \texttt{THEN} $r_+$ \texttt{ELSE} $r_-$,

\item[\bf parallel composition.] \texttt{FORALL} $x$ \texttt{WITH} $\varphi(x)$ \ $r(x)$,

\item[\bf bounded parallel composition.] $r_1 \dots r_n$,

\item[\bf choice.] \texttt{CHOOSE} $x$ \texttt{WITH} $\varphi(x)$ \texttt{IN} $r(x)$, and

\item[\bf let.] \texttt{LET} $x=t$ \texttt{IN} $r(x)$.

\end{description}\end{quote}

Each rule yields an update set $\Delta(S)$ in state $S$. If this update set is consistent, i.e. it does not contain two updates $(\ell,v), (\ell,v^\prime)$ with the same location $\ell$ and different values $v \neq v^\prime$, then applying this update set defines a successor state $S + \Delta(S)$.

Regarding the function symbols in $\Sigma$ we further distinguish static, dynamic and derived functions. Only dynamic function symbols can appear as the outermost symbol on the left hand side of an assignment.

\subsection{Concurrent Communicating Abstract State Machines}

A {\em concurrent ASM} (cASM) $\mathcal{CM}$ is defined as a family $\{ (a, \as_a) \}_{a \in \mathcal{A}}$ of pairs consisting of an agent $a$ and an ASM $\as_a$. 

Let $\Sigma_a$ denote the signature of the ASM $\as_a$. Taking the union $\Sigma = \bigcup_{a \in \mathcal{A}} \Sigma_a$ we distinguish between $\mathcal{CM}$-states built over $\Sigma$ and local states for agent $a$ built over $\Sigma_a$; the latter ones are simply projections of the former ones on the subsignature.

\begin{definition}\label{def-run}\rm

A {\em concurrent $\mathcal{CM}$-run} of a concurrent ASM $\{ (a, \as_a) \}_{a \in \mathcal{A}}$ is a sequence $S_0, S_1, S_2, \dots$ of $\mathcal{CM}$-states, such that for each $n \ge 0$ there is a finite set $A_n \subseteq \mathcal{A}$ of agents such that $S_{n+1}$ results from simultaneously applying update sets $\Delta_a(S_{j(a)})$ for all agents $a \in A_n$ that have been built by $\as_a$ in some preceding state $S_{j(a)}$ ($j(a) \leq n$ depending on~$a$), i.e. $S_{n+1} = S_n + \bigcup_{a \in A_n} \Delta_a(S_{j(a)})$ and $a \notin \bigcup_{i=j(a)}^{n-1} A_i$.

\end{definition}

Dynamic functions in $\Sigma_a$ can further be private or shared. In the latter case they can be updated also by other agents and thus appear in at least one other signature $\Sigma_{a^\prime}$.

In order to isolate agents responsible for a memory management subsystem we exploit {\em communicating concurrent ASMs} (ccASM) \cite{boerger:jucs2017}. In a ccASM the only shared function symbols take the form of mailboxes. Sending of a message $m$ from $a$ to $b$ means to update the out-mailbox of $a$ by inserting $m$ into it. This mailbox is a set-valued shared location with the restriction that only the sender can insert messages into it and only the environment---in this case understood as the message processing system---can read and delete them. The message processing system will move the message $m$ to the in-mailbox of the receiver $b$. Receiving a message $m$ by $b$ means in particular that $b$ removes $m$ from its in-mailbox and performs some local operation on $m$. 

Therefore, in ccASMs the language of ASM rules above is enriched by the following constructs\footnote{As explained in~\cite{boerger:jucs2017}, instead of describing the details of the local \texttt{RECEIVE} action of an agent we only use the corresponding  \texttt{RECEIVED}$(m)$ predicate; it means that $m$ is in the mailbox of the agent who (or in general whose mailbox manager) \texttt{RECEIVED} it (when the message processing system has inserted  $m$ into $b$'s in-mailbox and deleted $m$ from $a$'s out-mailbox).}:

\begin{quote}\begin{description}

\item[\bf Send.] \texttt{SEND}($\langle$message$\rangle$, from:$\langle$sender$\rangle$, to:$\langle$receiver$\rangle$),

\item[\bf Receive.] \texttt{RECEIVE}($\langle$message$\rangle$, from:$\langle$sender$\rangle$, to:$\langle$receiver$\rangle$), 

\item[\bf Received.] \texttt{RECEIVED}($\langle$message$\rangle$, from:$\langle$sender$\rangle$, to:$\langle$receiver$\rangle$), and

\item[\bf Consume.] \texttt{CONSUME}($\langle$message$\rangle$, from:$\langle$sender$\rangle$, to:$\langle$receiver$\rangle$).

\end{description}\end{quote}

Let us consider the situation, where all shared data is organised in relations with a unique primary key. We can model this by a set of function symbols $\Sigma_{\text{mem}} = \{ p_1 ,\dots, p_k \}$, where each $p_i$ has a fixed arity $a_i$, and a fixed ``co-arity'' $c_i$, such that in each state $S$ we obtain partial functions $p_i^S : B^{a_i} \rightarrow B^{c_i}$, which are almost everywhere undefined. 

When dealing with a memory management subsystem we have to specify a subsystem with separate agents that maintain the locations $(p_i,\vec{k})$. A read access by an agent $a \in \mathcal{A}$ aims at receiving a subset of relation $p_i$ containing those records with key values satisfying a condition $\varphi$, i.e. the subsystem has to evaluate a term of the form $p_i[\varphi] \;=\; \{ (\vec{k},\vec{v}) \mid \varphi(\vec{k}) \wedge \vec{v} \neq \textit{undef\/} \wedge p_i(\vec{k}) = \vec{v} \}$. As $p_i$ is not in the signature $\Sigma_a$, the agent $a$, instead of using the term $p_i[\varphi]$ in a rule, must send a read-request and wait for a response, i.e. it executes
\begin{asm}
\texttt{SEND}(read($p_i,\varphi$),from:$a$,to:\textit{\textit{home\/}}($a$))
\end{asm}

\noindent 
(the message just contains the name $p_i$ of the function symbol and the formula $\varphi$, and $\textit{home\/}(a)$ denotes a not further specified receiver, which for agent $a$ represents the memory management subsystem), and waits until, once \texttt{RECEIVED}($\text{answer}(p_i,\varphi)$,from:\textit{home\/}($a$),to:$a$) becomes true, it can
\begin{asm}
\texttt{RECEIVE}($\text{answer}(p_i,\varphi)$,from:\textit{home\/}($a$),to:$a$)
\end{asm}

\noindent 
the requested value from the memory management subsystem. We abstract from the details of the communication but assume the communication to be reliable (no message gets lost or damaged). Where clear from the context for reasons of succinctness we notationally omit the sender and receiver parameters in \texttt{SEND} and \texttt{RECEIVE}. 

Naturally, the received answer corresponds to a unique previously sent read-request. As in the sequel we concentrate on the handling of single message by the memory management subsystem, the correct association of possibly several answers to several requests is of minor importance for us.

The answer in the message must be a relation of arity $a_i + c_i$ satisfying the key property above. The agent can store such an answer using a non-shared function $p_i^a$ or process it in any other way, e.g. aggregate the received values. This is part of the ASM rule in $\as_a$, which we do not consider any further.

In the \texttt{SEND} and \texttt{RECEIVE} rules we use a fixed agent \textit{home\/}($a$), with which the agent $a$ communicates. It will be unknown to the agent $a$, whether this agent \textit{home\/}($a$) processes the read-request or whether it communicates with other agents to produce the answer.

Analogously, for bulk write access an agent $a$ may want to execute the operation $p_i \;\mbox{:\&}\; p$ to update all records\footnote{Note that in this way we capture a deletion of a record in $p_i$ with key $\vec{k}$ by having $(\vec{k},\textit{undef\/}) \in p$. Also insertions are subsumed by this operation: if $(\vec{k},\vec{v}) \in p$ holds, but $val_S(p_i,\vec{k}) = \textit{undef\/}$, then a new record in $p_i$ with key $\vec{k}$ is inserted.} with a key defined in $p$ to the new values given by $p$. While this would correspond to the ASM rule \texttt{FORALL} $(\vec{k},\vec{v}) \in p$ \ $p_i(\vec{k}) := \vec{v}$, the agent $a$ must send a write-request and wait for a response, i.e. it executes
\begin{asm}
\texttt{SEND}(write($p_i,p$),to:$\textit{home\/}(a)$)
\end{asm}

\noindent 
(again, the message only contains the name of the function symbol $p_i$ and a relation $p$), and waits to receive an acknowledgement to the write-request\footnote{Naturally, the remark above concerning the association of an answer to a unique previously sent request, extends analogously to write requests.}, i.e. to
\begin{asm}
\texttt{RECEIVE}(acknowledge(write,$p_i$),from:$\textit{\textit{home\/}\/}(a)$).
\end{asm}

We use the notation $\mathcal{CM}_0 = \{ (a, \as^c_a) \}_{a \in \mathcal{A}} \cup \{ (db, \as_{db}) \}$ for the concurrent communicating ASM. Here for the sake of completeness we may think of a single memory agent $db$---in particular, we have $\textit{home\/}(a) = db$ for all $a \in \mathcal{A}$---that receives read and write requests and executes them in one step\footnote{Note that the answer to a read request is a set, which may be empty.}. Thus, the rule of $\as_{db}$ looks as follows:
\begin{asm}
\texttt{IF} \> \texttt{RECEIVED}(read($p_i,\varphi$),from:$a$) \\
\texttt{THEN} \\
\> \texttt{CONSUME}(read($p_i,\varphi$),from:$a$) \\
\> \texttt{LET} answer($p_i,\varphi$) $= \{ (\vec{k},\vec{v}) \mid \varphi(\vec{k}) \wedge p_i(\vec{k}) = \vec{v} \wedge \vec{v} \neq \textit{undef\/} \}$ \texttt{IN} \\
\>\> \texttt{SEND}(answer($p_i,\varphi$),to:$a$) \\
\texttt{IF} \> \texttt{RECEIVED}(write($p_i,p$),from:$a$) \\
\texttt{THEN} \\
\> \texttt{CONSUME}(write($p_i,p$),from:$a$) \\
\> \texttt{FORALL} $(\vec{k},\vec{v}) \in p$ \\
\>\> $p_i(\vec{k}) := \vec{v}$ \\
\> \texttt{SEND}(acknowledge(write,$p_i$),to:$a$)
\end{asm}

\subsection{Memory Organisation with Replication}

In the previous subsection we assumed that logically (from the users' point of view) all data is organised in relations $p_i$ with a unique primary key. However, as we want to emphasise {\em replication}, instead of a location $(p_i,\vec{k})$ there will always be several replicas, and at each replica we may have a different value. We use the notion {\em cluster} to refer to all (replicated) locations associated with a logical relation $p_i$.

Each cluster is associated with several {\em data centres}, and each data centre comprises several {\em nodes}. The nodes are used for data storage, and data centres correspond to physical machines maintaining several such storage locations. So let $\mathcal{D}$ denote the set of data centres, and let $\mathcal{D}_i$ ($i=1,\dots,k$) be the sets of data centres for maintaining the relations $p_i$, i.e. $\mathcal{D} = \bigcup_{i=1}^k \mathcal{D}_i$.

First let us assume that each relation $p_i$ is fragmented according to the values of a hash-key. That is, for each $i = 1,\dots,k$ we can assume a static hash-function $h_i : B^{a_i} \rightarrow [m,M]
\subseteq \mathbb{Z}$ assigning a hash-key to each key value in $B^{a_i}$, i.e. those keys at which in the memory management system possibly some value may be defined. We further assume a partition
$[m,M] = \bigcup_{j=1}^{q_i} \textit{range}_j$ of the interval of hash-key values such that $\textit{range}_{j_1} < \textit{range}_{j_2}$ holds for all $j_1 < j_2$, so each range will
again be an interval. These range intervals are used for the horizontal fragmentation into $q_i$ fragments of the to-be-represented function $p_i$: $\textit{Frag}_{j,i} = \{ \vec{k} \in B^{a_i}\mid h_i(\vec{k}) \in \textit{range}_j \}$.

All these fragments will be replicated and their elements associated with a value (where defined by the memory management system), using a fixed {\em replication factor} $r_i$ for each cluster. That is, each fragment $\textit{Frag}_{j,i}$ will be replicated $r_i$-times for each data centre $d \in \mathcal{D}_i$. A set of all pairs $(\vec{k},\vec{v})$ with key $\vec{k} \in \textit{Frag}_{j,i}$ and an associated value $\vec{v}$ in the memory management system is called a replica of $\textit{Frag}_{j,i}$.

More precisely assume that each data centre $d \in \mathcal{D}_i$ consists of $n_i$ nodes, identified by $d$ and a number $j^\prime \in \{ 1,\dots,n_i \}$. Then we use a predicate $\textit{copy\/}(i,j,d,j^\prime)$ to denote that the node with number $j^\prime$ in the data centre $d \in \mathcal{D}_i$ contains a replica of $\textit{Frag}_{j,i}$. To denote the values in replicas we use dynamic functions $p_{i,j,d,j^\prime}$ of arity $a_i$ and co-arity $c_i + 1$ (functions we call again replicas). That is, instead of the logical function symbol $p_i$ we use function symbols $p_{i,j,d,j^\prime}$ with $j \in \{ 1,\dots,q_i \}$, $d \in \mathcal{D}_i$ and $j^\prime \in \{ 1,\dots,n_i \}$, and we request $h_i(\vec{k}) \in range_j$ for all $\vec{k} \in B^{a_i}$ whenever $copy\/(i,j,d,j^\prime)$ holds and $p_{i,j,d,j^\prime}(\vec{k})$ is defined. For the associated values we have $p_{i,j,d,j^\prime}(\vec{k}) = (\vec{v},t)$, where $t$ is an added timestamp value, and values $\vec{v}$ may differ from replica to replica, i.e. there can be different values $\vec{v}$ with different timestamps in different replicas of $\textit{Frag}_{j,i}$.

Each data centre $d$ maintains a logical clock $clock_d$ that is assumed to advance (without this being further specified), and $clock_d$ evaluates to the current time at data centre $d$. Timestamps must satisfy the following requirements:

\begin{enumerate}

\item Timestamps are totally ordered.

\item Timestamps set by different data centres are different from each other\footnote{This requirement can be fulfilled by assuming that a timestamp created by data centre $d$ has the form $n + o_d$ with a positive integer $n$ and an offset $o_d \in (0,1)$, such that offsets of different data centres are different. Equivalently, one might use integer values for timestamps plus a total order on data centres that is used for comparisons in case two timestamps are otherwise equal.}.

\item Timestamps respect the inherent order of message passing, i.e. when data with a time\-stamp $t$ is created at data centre $d$ and sent to data centre $d^\prime$, then at the time the message is received the clock at $d^\prime$ must show a time larger than $t$. 
\label{condition2}

\end{enumerate}

When the condition \ref{condition2} is not met, a data centre may also adjust its clock for logical time synchronisation according to Lamport's algorithm in \cite{lamport:cacm1978}. For clock adjustment let us define $\text{adjust\_clock}(d,t) \equiv clock_d := t^\prime$, where $t^\prime$ is the smallest possible timestamp at data centre $d$ with $t \le t^\prime$.

\subsection{Internal Request Handling for Replicated Memory}

When dealing with a memory management subsystem with replication the request messages sent by agents $a$ remain the same, but the internal request handling by the memory management subsystem changes. This will define a refined concurrent communicating ASM $\mathcal{CM}_1 = \{ (a,\as^c_a) \}_{a \in \mathcal{A}} \cup  \{ (d,\as_d) \}_{d \in \mathcal{D}}$.

Each agent $a$ possesses a private version of the shared location $(p_i,\vec{k})$ parameterised by $a$ itself\footnote{An exact definition for this is given by the notion of {\em ambient dependent function} in \cite{boerger:jcss2012}.}. Consider a read request $\text{read}(p_i,\varphi)$ received from agent $a$ by data centre $d$---let $i$ be fixed for the rest of this section. Due to the fact that data is horizontally fragmented, we need to evaluate several requests $\text{read}(p_{i,j},\varphi)$ concerning keys $\vec{k}$ with $h_i(\vec{k}) \in \textit{range}_j$, one request for each fragment index $j$, and then build the union so that $p_i[\varphi] = \bigcup_{j=1}^{q_i} p_{i,j}[\varphi]$.

In order to evaluate $p_{i,j}[\varphi]$ several replicas of $\textit{Frag}_{j,i}$ will have to be accessed. Here we will leave out any details on how these replicas will be selected and accessed---this will be handled later by means of refinement. We only request that the selection of replicas complies with a
\textit{read-policy\/}. Such a policy will also be left abstract for the moment and defined later\footnote{We also leave out the treatment of nodes that are not reachable. It can be tacitly assumed that if node $j$ at data centre $d$ cannot be reached, then all operations affecting data at this node will be logged and executed once the node becomes available again.}.

When reading actual data, i.e. evaluating $p_{i,j,d,j^\prime}(\vec{k})$ for selected key values $\vec{k}$, we obtain different time-stamped values $(\vec{v},t)$, out of which a value $\vec{v}$ with the latest
timestamp is selected and sent to $a$ as the up-to-date value of $p_i(\vec{k})$\footnote{When there is no record in relation $p_{i,j,d,j^\prime}$ with key $\vec{k}$, we would normally write $p_{i,j,d,j^\prime}(\vec{k}) = \textit{undef\/}$ in an ASM, but in the replication context it will be simpler to write $p_{i,j,d,j^\prime}(\vec{k}) = (\textit{undef\/},-\infty)$ instead, i.e. all non-existing data are considered to carry the smallest possible timestamp denoted by $-\infty$. Furthermore, if a record is deleted, we keep a deletion timestamp, so we may also find $p_{i,j,d,j^\prime}(\vec{k}) = (\textit{undef\/},t)$ with $t > -\infty$. As we will see, such a deletion timestamp $t$ becomes obsolete, once the value $(\textit{undef\/},t)$ has been assigned to all replicas, i.e. the assumption of deletion timestamps does not disable physical removal of data from the database.}. The requirement that timestamps set by different data centres differ implies that for given $\vec{k}$ the value $\vec{v}$ with the latest timestamp is unique. All records obtained this way will be returned as the result of the read request
to the issuing agent $a$. Thus, we obtain the following ASM rule \texttt{AnswerReadReq} to-be-executed by any data centre $d$ upon receipt of a read request from an agent $a$:

\begin{asm}
\texttt{AnswerReadReq} = \\
\texttt{IF} 
        \texttt{RECEIVED}($\text{read}(p_i,\varphi),\text{from:}a$) \texttt{THEN} \\
\> \texttt{CONSUME}($\text{read}(p_i,\varphi),\text{from:}a$)\\
\> \texttt{FORALL} $j \in \{ 1,\dots,q_i \}$ 
            \texttt{CHOOSE} $G_{i,j}$ \texttt{WITH} \\
\>\;\; $G_{i,j} \subseteq 
      \{ (d^\prime,j^\prime) \mid \textit{copy\/}(i,j,d^\prime,j^\prime) \wedge
      d^\prime \in \mathcal{D}_i \wedge 1 \leq j^\prime \leq n_i\}$\\
\>\> $\wedge \textit{complies\/}(G_{i,j},\text{read-policy})$ \\
\> \texttt{LET} $t_{\max}(\vec{k}) = \max \{ t \mid \exists \vec{v}^\prime, \bar{d}, \bar{j}. 
       (\bar{d},\bar{j}) \in G_{i,j} \wedge
        p_{i,j,\bar{d},\bar{j}}(\vec{k}) = (\vec{v}^\prime,t) \} $ \texttt{IN}\\
\> \texttt{LET} $\text{answer}_{i,j} =  \{ (\vec{k},\vec{v}) \mid 
     \varphi(\vec{k}) \wedge h_i(\vec{k}) \in \textit{range}_j \wedge \vec{v} \neq \textit{undef\/} \wedge$ \\
\>\> $\exists d^\prime, j^\prime . ( (d^\prime,j^\prime) \in G_{i,j} \wedge 
     p_{i,j,d^\prime,j^\prime}(\vec{k}) = (\vec{v},t_{\max}(\vec{k}))) \}$ \texttt{IN}\\
\> \texttt{LET} $\text{answer}(p_i,\varphi) = \bigcup_{j=1}^{q_i} \text{answer}_{i,j}$ \texttt{IN} \\
\>\>         \texttt{SEND}($\text{answer}(p_i,\varphi),\text{to:}a$)
\end{asm}

Note that the unique value $\vec{v}$ with $p_{i,j,d^\prime,j^\prime}(\vec{k}) = (\vec{v},t_{\max}(\vec{k}))$ may be \textit{undef\/} and that the returned $\text{answer}(p_i,\varphi)$ may be the empty set.

For a write request write($p_i,p$) sent by agent $a$ to data centre $d$ we proceed analogously. In all replicas of $\textit{Frag}_{j,i}$ selected by a {\em write-policy} the records with a key value in $p$ will be updated to the new value---this may be \textit{undef\/} to capture deletion---provided by $p$, and a timestamp given by the current time $\textit{clock}_d$. However, the update will not be executed, if the timestamp of the existing record is already newer. In addition, clocks that ``are too late'' will be adjusted, i.e. if the new timestamp received from the managing data centre $d$ is larger than the timestamp at data centre $d^\prime$, the clock at $d^\prime$ is set to the received timestamp. Thus, we obtain the following ASM rule \texttt{PerformWriteReq}\footnote{Again, we dispense with the handling of write-requests at nodes that are not available. For this we can assume an exception handling procedure that logs requests and executes them in the order of logging, once a node has become alive again. This could give rise to a refinement, which we omit here.} to-be-executed by any data centre $d$ upon receipt of an update request from an agent $a$:

\begin{asm}
\texttt{PerformWriteReq} = \\
\texttt{IF} \texttt{RECEIVED}($\text{write}(p_i,p),\text{from:}a$) \texttt{THEN} \\
\;\;\; \texttt{CONSUME}($\text{write}(p_i,p),\text{from:}a$)\\
\;\;\; \texttt{FORALL} $j \in \{ 1,\dots,q_i \}$  
               \texttt{CHOOSE} $G_{i,j}$ \texttt{WITH} \\
      \> $G_{i,j} \subseteq 
      \{ (d^\prime,j^\prime) \mid \textit{copy\/}(i,j,d^\prime,j^\prime) \wedge
      d^\prime \in \mathcal{D}_i \wedge 1 \leq j^\prime \leq n_i\}$\\
      \> $\wedge \textit{complies\/}(G_{i,j},\text{write-policy})$ \\
\;\;\; \texttt{LET} $t_{\text{current}} = \textit{clock}_{\texttt{self}}$ \texttt{IN}\\
\> \texttt{FORALL} $(d^\prime,j^\prime) \in G_{i,j}$ \\
\>\;\;\;\texttt{FORALL} $(\vec{k},\vec{v}) \in p$ 
          \texttt{WITH} $h_i(\vec{k}) \in \textit{range}_j$  \\
\>\> \texttt{IF} $\exists \vec{v}^\prime, t . 
  p_{i,j,d^\prime,j^\prime}(\vec{k}) = (\vec{v}^\prime, t) \wedge t < t_{\text{current}}$ \texttt{THEN}\\
\>\>\;\;\; $p_{i,j,d^\prime,j^\prime}(\vec{k}) := (\vec{v},t_{\text{current}})$\\
\>\;\;\;\texttt{IF}  $clock_{d^\prime} < t_{\text{current}}$
\texttt{THEN} adjust\_clock$(d^\prime,t_{\text{current}})$ \\
\;\;\; \texttt{SEND}(acknowledge($\textit{write\/},p_i),\text{to:}a$)
\end{asm}

The clock adjustment is necessary to ensure that timestamps respect the inherent order of message passing as requested above. Write requests with old timestamps may be lost in case a value with a newer timestamp already exists. Then depending on the read-policy a lost update on a single replica may be enough for the value never to appear in an answer to a read-request. In the following we use the notions of {\em complete} and {\em correct} refinement\footnote{Note that the notion of refinement for ASMs is more general than data refinement as discussed in \cite[p.113]{boerger:2003}. In particular, correct refinement does not imply the preservation of invariants.} as defined in \cite[pp.111ff.]{boerger:2003}.

\begin{proposition}\label{prop-1st-refinement}

The concurrent communicating ASM $\mathcal{CM}_1 = \{ (a,\as^c_a) \}_{a \in \mathcal{A}} \cup  \{ (d,\as_d) \}_{d \in \mathcal{D}}$ is a complete refinement of the concurrent communicating ASM $\mathcal{CM}_0 = \{ (a,\as^c_a) \}_{a \in \mathcal{A}} \cup  \{ (db,\as_{db}) \}$.

\end{proposition}

\begin{proof}

The only differences between the two communicating concurrent ASMs are the following:

\begin{itemize}

\item Instead of having $\textit{home\/}(a) = db$ for all $a \in \mathcal{A}$ in the abstract specification,  $\textit{home\/}(a) \in \mathcal{D}$ differs in the refinement. Nonetheless, in both cases the handling of a read or write request is done in a single step.

\item The rule fragment in the abstract specification

\begin{asm}
\> \texttt{LET} answer($p_i,\varphi$) $= \{ (\vec{k},\vec{v}) \mid \varphi(\vec{k}) \wedge p_i(\vec{k}) = \vec{v} \wedge \vec{v} \neq \textit{undef\/} \}$ \texttt{IN} \\
\>\> \texttt{SEND}(answer($p_i,\varphi$),to:$a$)
\end{asm}

dealing with a read request corresponds to a rule fragment

\begin{asm}
\> \texttt{FORALL} $j \in \{ 1,\dots,q_i \}$ 
            \texttt{CHOOSE} $G_{i,j}$ \texttt{WITH} \\
\>\;\; $G_{i,j} \subseteq 
      \{ (d^\prime,j^\prime) \mid \textit{copy\/}(i,j,d^\prime,j^\prime) \wedge
      d^\prime \in \mathcal{D}_i \wedge 1 \leq j^\prime \leq n_i\}$\\
\>\> $\wedge \textit{complies\/}(G_{i,j},\text{read-policy})$ \\
\> \texttt{LET} $t_{\max}(\vec{k}) = \max \{ t \mid \exists \vec{v}^\prime, \bar{d}, \bar{j}. 
       (\bar{d},\bar{j}) \in G_{i,j} \wedge
        p_{i,j,\bar{d},\bar{j}}(\vec{k}) = (\vec{v}^\prime,t) \} $ \texttt{IN}\\
\> \texttt{LET} $\text{answer}_{i,j} =  \{ (\vec{k},\vec{v}) \mid 
     \varphi(\vec{k}) \wedge h_i(\vec{k}) \in \textit{range}_j \wedge \vec{v} \neq \textit{undef\/} \wedge$ \\
\>\> $\exists d^\prime, j^\prime . ( (d^\prime,j^\prime) \in G_{i,j} \wedge 
     p_{i,j,d^\prime,j^\prime}(\vec{k}) = (\vec{v},t_{\max}(\vec{k}))) \}$ \texttt{IN}\\
\> \texttt{LET} $\text{answer}(p_i,\varphi) = \bigcup_{j=1}^{q_i} \text{answer}_{i,j}$ \texttt{IN} \\
\>\>         \texttt{SEND}($\text{answer}(p_i,\varphi),\text{to:}a$)
\end{asm}

in the rule \texttt{AnswerReadReq} in the refinement.

\item The rule fragment in the abstract specification

\begin{asm}
\> \texttt{FORALL} $(\vec{k},\vec{v}) \in p$ \ $p_i(\vec{k}) := \vec{v}$ \\
\> \texttt{SEND}(acknowledge(write,$p_i$),to:$a$)
\end{asm}

dealing with a write request corresponds to a rule fragment

\begin{asm}
\> \texttt{FORALL} $j \in \{ 1,\dots,q_i \}$  
               \texttt{CHOOSE} $G_{i,j}$ \texttt{WITH} \\
      \> $G_{i,j} \subseteq 
      \{ (d^\prime,j^\prime) \mid \textit{copy\/}(i,j,d^\prime,j^\prime) \wedge
      d^\prime \in \mathcal{D}_i \wedge 1 \leq j^\prime \leq n_i\}$\\
      \> $\wedge \textit{complies\/}(G_{i,j},\text{write-policy})$ \\
\> \texttt{LET} $t_{\text{current}} = \textit{clock}_{\texttt{self}}$ \texttt{IN}\\
\> \texttt{FORALL} $(d^\prime,j^\prime) \in G_{i,j}$ \\
\>\> \texttt{FORALL} $(\vec{k},\vec{v}) \in p$ 
          \texttt{WITH} $h_i(\vec{k}) \in \textit{range}_j$  \\
\>\> \texttt{IF} $\exists \vec{v}^\prime, t . 
  p_{i,j,d^\prime,j^\prime}(\vec{k}) = (\vec{v}^\prime, t) \wedge t < t_{\text{current}}$ \texttt{THEN}\\
\>\>\> $p_{i,j,d^\prime,j^\prime}(\vec{k}) := (\vec{v},t_{\text{current}})$\\
\>\> \texttt{IF}  $clock_{d^\prime} < t_{\text{current}}$
\texttt{THEN} adjust\_clock$(d^\prime,t_{\text{current}})$ \\
\> \texttt{SEND}(acknowledge($\textit{write\/},p_i),\text{to:}a$)
\end{asm}

in the rule \texttt{PerformWriteReq} in the refinement.

\end{itemize}

Thus, each run of the abstract communicating concurrent ASM defines in a natural way a run of the refined communicating concurrent ASM\footnote{Actually, in this case the refinement is a (1,1)-refinement.}. \qed

\end{proof}

Note that without further knowledge about the read- and write-policies it is not possible to prove that the refinement is also correct.

\section{Refinement Using Replication Policies}\label{sec:consistency}

We define {\em view compatibility}, a notion for consistency that formalises the intuitive expectation of the agents sending requests that the answers in case of replication remain the same as without, because replication is merely an internal mechanism of the memory management subsystem to increase availability, which is completely hidden from the requesting agents. We then refine the ASMs for internal request handling by concrete read- and write-policies, and show that for particular combinations of read- and write-policies our abstract specification $\mathcal{CM}_1$ guarantees view compatibility, which further implies that the refinement of $\mathcal{CM}_0$ by $\mathcal{CM}_1$ is correct.

\subsection{View Compatibility}

Informally, view compatibility is to ensure that the system behaves in a way that whenever an agent sends a read- or write-request the result is the same as if the read or write had been executed in a state without replication or timestamps and without any internal processing of the request\footnote{Note that view compatibility is a rather weak consistency requirement, as it only ensures that despite replication up-to-date values are read. However, as we will show later in Section \ref{sec:data}, even this weak concistency requirement requires some form of transaction management.}. However, it may be possible that parallel requests are evaluated in different states.

For a formal definition of this notion we have to relate runs of the concurrent communicating ASM with the memory management subsystem with runs, where in each state a virtual location $(p_i,\vec{k})$ has only one well-defined value $\vec{v}$ instead of computing such a value from the different replicas that may even originate from different states. For this we first introduce the technical notion of a {\em flattening}: we simply reduce the multiple values associated with replicas of a location $\ell$ to a single value.

Formally, consider runs of the concurrent ASM $\mathcal{CM}_1 = \{ (a,\as^c_a) \}_{a \in \mathcal{A}} \cup \{ (d,\as_d) \}_{d \in \mathcal{D}}$, where the ASMs $\as^c_a$ sends read and write requests that are processed by the data centre agents $d \in \mathcal{D}$, which return responses. 

\begin{definition}\rm

If $S_0, S_1, S_2 , \dots$ is a run of $\mathcal{CM}_1 = \{ (a,\as^c_a) \}_{a \in \mathcal{A}} \cup \{ (d,\as_d) \}_{d \in \mathcal{D}}$, then we obtain a {\em flattening} $S_0^\prime, S_1^\prime, S_2^\prime$, \dots\ by replacing in each state all locations $(p_{i,j,d^\prime,j^\prime},\vec{k})$ by a single location $(p_i,\vec{k})$ and letting the value associated with $(p_i,\vec{k})$ in the considered state be one of the values in the set
\[ 
\{ \vec{v} \mid \exists j,d^\prime,j^\prime . \exists t . p_{i,j,d^\prime,j^\prime}(\vec{k}) = (\vec{v},t) \} 
\]

\noindent 
in the considered state. 

\end{definition}

Obviously, a flattening is a sequence of states of the concurrent ASM $\{ (a,\as^c_a) \}_{a \in \mathcal{A}}$, but in most cases it will not be a run. So the question is, under which conditions we can obtain a run. As we stated above that the system shall behave as if there is no replication, we have to ensure that for each agent the following property holds: If the agent sends a write-request that would update the value of a location $\ell$ to $\vec{v}$, the effect would be the same in the flattening. Analogously, if the agent sends a read-request for location $\ell$, the answer it receives would be also the same as if a read is executed in the flattening.

In order to formalise this property we define agent views, which emphasise only the moves of a single agent $a \in \mathcal{A}$, while merging all other moves into activities of the environment as defined in \cite{boerger:2003}. 

Let us first consider the {\em agent view} of a concurrent run $S_0, S_1, \dots$ of our communicating ASMs $\{ (a,\as^c_a) \}_{a \in \mathcal{A}}$. Let $a \in \mathcal{A}$ be an arbitrary agent. Its view of the run is the subsequence of states $S_{a,0}, S_{a,1}, \dots$ in which $a$ makes a move (restricted to the signature of $a$). Thus $S_{a,0}$ is $a$'s initial state, the state $S_j$ (where $j$ depends on $a$) in which $a$  performs its first step in the given run. Given any state $S_k=S_{a,n}$, its successor state in the $a$-view sequence depends on the move $a$ performs in $S_{a,n}$. 

\begin{enumerate}
	
\item If $a$ in $S_{a,n}$ performs a Send step---a write- or read-request to $\textit{\textit{home\/}\/}(a)$---it contributes to the next state $S_{k+1}$ by an update set which includes an update of its out-mailbox, which in turn is assumed to eventually yield an update of the mailbox of $\textit{\textit{home\/}\/}(a)$.  But $S_{k+1}$ is not yet the next $a$-view state, in which $a$ will perform its next move. This move is determined by the following assumption:
	
\emph{Atomic Request/Reply Assumption} for agent/db runs: If in a run an agent performs a Send step to the memory management system, then its next step in the run is the corresponding Receive step, which can be performed once the answer to the query sent by the memory management system has been Received. 
	
By this assumption the next $a$-view state $S_{a,n+1}=S_{l}$ is determined by (what appears to $a$ as) an environment action which enables the Receive step by inserting the reply message into $a$'s mailbox and thereby making the $Received$ predicate true in $S_l$  for some $l>k$.                                                                                                                                                                                         
	
\item  If $a$ in $S_{a,n}=S_k$ performs a Receive or an internal step, then besides the mailbox update to Consume the received message it yields only updates to non-shared locations so that its next $a$-view state is the result of applying these updates together with updates other agents bring in to form $S_{k+1}=S_{a,n+1}$. 
	
\end{enumerate}

Note that by the \emph{Atomic Request/Reply Assumption} any agent can make finitely many internal steps after and only after each Receive step. For the ease of exposition but without loss of generality we synchronize internal steps of an agent with the global steps other agents perform during such an internal computation segment so that the result of internal moves becomes visible in the global run view. 
 
We now define a notion of \emph{flat agent view} of a run $S_0, S_1, \dots$ of the concurrent agent ASM $\{ (a,\as^c_a) \}_{a \in \mathcal{A}} \cup \{ (d,\as_d) \}_{d \in \mathcal{D}}$, including the memory management subsystem.

Take an arbitrary subsequence $S_{j_0}^\prime, S_{j_1}^\prime, \dots$ of an arbitrary flattening $S_0^\prime, S_1^\prime, \dots$ (restricted to the signature of the agent $a$) of $S_0, S_1, \dots$. Then $S_{j_0}^\prime, S_{j_1}^\prime, \dots$ is called a {\em flat view} of agent $a$ of the run $S_0, S_1, \dots$ if the following conditions hold:

\begin{itemize}
	
\item Whenever $a$ performs a request in state $S_k$ there is some $S_{j_i}^\prime$ such that $k=j_i$. If the corresponding reply is received in state $S_m$ for some $m>k$ (so that $a$ makes a Receive move in state $S_m$), then $S_{j_{i+1}}^\prime = S_m$. Furthermore there exists some $n$ with $k < n\leq m$ such that the following holds:

\begin{itemize}
   	
\item If the request is a write-request, then for each location $\ell$ with value $v$ in this request $val_{S_n^\prime}(\ell) = v$ holds, provided there exists an agent reading\footnote{Otherwise the update at location $\ell$ will be lost. For instance, this may happen in case of outdated timestamps.} the value $v$.
   	
\item If the request is a read-request, then for each location $\ell$ with value $v$ in the answer $val_{S_n^\prime}(\ell) = v$ holds.

\end{itemize}

\item Whenever $a$ performs a Receive or an internal move in state $S_k$ there is some $j_i$ such that $S_k=S_{j_i}^\prime$ and $S_{k+1}=S_{j_{i+1}}^\prime$

\end{itemize}

\begin{definition}\rm

We say that $\{ (a,\as^c_a) \}_{a \in \mathcal{A}} \cup \{ (d,\as_d) \}_{d \in \mathcal{D}}$ is {\em view compatible} with the concurrent ASM $\{ (a,\as^c_a) \}_{a \in \mathcal{A}} \cup \{ (db,\as_{db}) \}$ iff for each run $\mathcal{R} = S_0, S_1, S_2 , \dots$ of 
$\{ (a,\as^c_a) \}_{a \in \mathcal{A}} \cup \{ (d,\as_d) \}_{d \in \mathcal{D}}$ there exists a subsequence of a flattening $\mathcal{R}^\prime = S_0^\prime, S_1^\prime, S_2^\prime, \dots$ that is a run of $\{ (a,\as^c_a) \}_{a \in \mathcal{A}} \cup \{ (db,\as_{db}) \}$ such that for each agent $a \in \mathcal{A}$ the agent $a$-view of $\mathcal{R}^\prime$ coincides with a flat view of $\mathcal{R}$ by $a$.

\end{definition}

Note that in this definition of view compatibility we relate agent views with flat agent views, and in both these views we consider the restriction to the signature of the agent $a$. Technically, the mailboxes associated with $a$ belong to this signature, so every \texttt{SEND} and \texttt{RECEIVE} state $S_i$ (for $a$) and its flattening $S_i^\prime$ appears in the views. We can understand the agent view and the flat agent view in such a way, that the environment reads the request (i.e. the message in the out-mailbox) in state $S_{j_i}^\prime$, evaluates it in state $S_n^\prime$, and places the answer into the in-mailbox of $a$ in state $S_{j_{i+1}}^\prime$ (for $j_i < n < j_{i+1}$ as requested in the definition). This can also formally be considered as a move by a single agent $m$ representing the memory management subsystem, which executes the ASM rules that are associated with each request.

The notion of view compatibility is closely related to {\em sequential consistency} as defined by Lamport. In Lamport's work an execution is a set of sequences of operations, one sequence per 
process. An operation is a pair (operation, value), where the operation is either read or write and value is the value read or written. The execution is {\em sequentially consistent} iff there exists an interleaving of the set of sequences into a single sequence of operations that is a legal execution of a single-processor system, meaning that each value read is the most recently written value. (One also needs to specify what value a read that precedes any write can obtain.) Our definition of view compatibility generalises this notion, as we always consider bulk requests and also permit parallel operations.

\subsection{Specification of Replication Policies}

In the specification of ASM rules handling read and write
requests by a fixed data centre $d$ we used sets $G_{i,j} \subseteq C_{i,j}$
with $C_{i,j} = \{ (d^\prime,j^\prime) \mid
\textit{copy\/}(i,j,d^\prime,j^\prime) \}$ ($i,j$ are given by the
request and the fragmentation) as well as an abstract predicate
$\textit{complies\/}(G_{i,j},\text{policy})$. It is allowed to specify
different policies for read and write, but only one policy is used for
all read requests, and only one for all write requests.

Let us now define different such policies and refine the ASMs for the
data centre agents $d$ by means of these definitions. Basically, we
distinguish policies that are global and those that are local, the
latter ones requesting that only replicas in the handling data centre
$d$ are considered for the set $G_{i,j}$. Thus we can use the following
definition:
\[ \textit{local\/} \;\equiv\; \exists d . \forall d^\prime, j^\prime . 
(d^\prime,j^\prime) \in G_{i,j} \Rightarrow d^\prime = d \]

In addition the different policies differ only in the number of
replicas that are to be accessed. Major global policies are \textsc{All}, \textsc{One}, \textsc{Two},
\textsc{Three}, \textsc{Quorum}, and \textsc{Each\_Quorum}, while the
corresponding local policies are \textsc{Local\_One} and
\textsc{Local\_Quorum}.

\begin{description}

\item[\rm\textsc{All}.] As the name indicates, the predicate $\textit{complies\/}(G_{i,j},\textsc{All})$
  can be defined by $G_{i,j} = C_{i,j}$, i.e. all replicas are to be accessed.

\item[\rm\textsc{One}, \textsc{Two}, \textsc{Three}.] Here, at least one, two or
  three replicas are to be accessed, which defines
  $\textit{complies\/}(G_{i,j},\textsc{One})$,
  $\textit{complies\/}(G_{i,j},\textsc{Two})$ and
  $\textit{complies\/}(G_{i,j},\textsc{Three})$ by $|G_{i,j}| \geq 1$, $|G_{i,j}| \geq
  2$, and $|G_{i,j}| \geq 3$, respectively.

\item[\rm\textsc{Local\_One}.] Analogously,
  $\textit{complies\/}(G_j,\textsc{Local\_One})$ is defined by the conjunction
  $\textit{local\/} \;\wedge\; \textsc{ONE}$.

\item[\rm\textsc{Quorum}($q$).] For a value $q$ with $0 < q < 1$
  $\textit{complies\/}(G_{i,j},\textsc{Quorum}(q))$ is defined by $q \cdot
  |C_{i,j}| < |G_{i,j}|$. By default, the value $q = \frac{1}{2}$ is
  used, i.e. a majority of replicas has to be accessed.
  Note that \textsc{All} could be replaced by \textsc{Quorum}($1$).

\item[\rm\textsc{Local\_Quorum}($q$).] Analogously, 
  $\textit{complies\/}(G_{i,j},\textsc{Local\_Quorum}(q))$ can be defined by
  $\textit{local\/} \;\wedge\; \textsc{Quorum}(q)$.

\item[\rm\textsc{Each\_Quorum}($q$).] For this we have to consider each data
  centre separately, for which we need $C_{i,j,\bar{d}} = \{
  (d^\prime,j^\prime) \in C_{i,j} \mid d^\prime = \bar{d} \}$ and
  $G_{i,j,\bar{d}} = G_{i,j} \cap C_{i,j,\bar{d}}$. Then the definition of
  $\textit{complies\/}(G_{i,j},\textsc{Each\_Quorum}(q))$ becomes $\forall
  \bar{d} \in \mathcal{D}_i . q \cdot |C_{i,j,\bar{d}}| <
  |G_{i,j,\bar{d}}|$.

\end{description}

\subsection{Consistency Analysis}

In the following we will analyse in more detail the effects of the replication policies. For this we need {\em appropriate combinations} of read- and write-policies:

\begin{itemize}

\item If the write policy is \textsc{ALL}, then the combination with any read policy is appropriate.

\item If the read policy is \textsc{ALL}, then the combination with any write policy is appropriate.

\item If the write-policy is \textsc{Quorum}$(q)$ or \textsc{Each\_Quorum}$(q)$ and the read-policy is \textsc{Quorum}$(q^\prime)$ or \textsc{Each\_Quorum}$(q^\prime)$ with $q + q^\prime \ge 1$, then the combination is appropriate.

\end{itemize}

\begin{proposition}\label{prop:1}

Let $\mathcal{CM}_1 = \{ (a,\as^c_a) \}_{a \in \mathcal{A}} \cup \{ (d,\as_d) \}_{d \in \mathcal{D}}$ be the concurrent communicating ASM with a memory management subsystem using data centres $d \in \mathcal{D}$ as specified in the previous section. If the combination of the read and write policies is appropriate, then the system is view compatible with the concurrent ASM $\{ (a,\as^c_a) \}_{a \in \mathcal{A}}$.

\end{proposition}

\begin{proof}

We exploit that in our abstract specification of the handling of write- and read-requests all selected replicas are written and read in parallel. Thus, if the write-policy is \textsc{All}, then in each state always the same value is stored in all replicas, so the proposition is obvious for the policy \textsc{All}.

If the write-policy is \textsc{Quorum}$(q)$ or \textsc{Each\_Quorum}$(q)$, then for each
location $\ell$ the multiplicity of replica considered to determine the value with the largest timestamp is at least $\lceil \frac{m+1}{2} \rceil$ with $m$ being the total number of copies. For this it is essential that updates with a smaller timestamp are rejected, if a value with a larger timestamp already exists.

Consequently, each read access with one of the policies 
\textsc{Quorum}$(q^\prime)$, \textsc{Each\_Quorum}$(q^\prime)$ (with $q + q^\prime \ge 1$) or \textsc{All} will read at least once this value and return it, as the value with
the largest timestamp will be used for the result. That is, in every state only the value
with the largest timestamp for each location uniquely
determines the run, which defines the equivalent
concurrent run.\qed

\end{proof}

\begin{corollary}\label{cor-correct}

If $\mathcal{CM}_1 =  \{ (a,\as^c_a) \}_{a \in \mathcal{A}} \cup  \{ (d,\as_d) \}_{d \in \mathcal{D}}$ is view compatible, then it is also a correct refinement of $\mathcal{CM}_0 =  \{ (a,\as^c_a) \}_{a \in \mathcal{A}} \cup  \{ (db,\as_{db}) \}$.

\end{corollary}

\begin{proof}

\ Let $\mathcal{R} = S_0, S_1, \dots$ be a run of $\mathcal{CM}_1$. According to the definition of view compatibility there exists a flattened subrun $\mathcal{R}^\prime = S_{i_0}^\prime, S_{i_1}^\prime, \dots$ that is also a concurrent run of $\mathcal{CM}_0 =  \{ (a,\as^c_a) \}_{a \in \mathcal{A}} \cup  \{ (db,\as_{db}) \}$ such that for each agent $a \in \mathcal{A}$ the projections of $\mathcal{R}$ and $\mathcal{R}^\prime$ coincide.\qed

\end{proof}

A stronger notion of consistency would be {\em global consistency}, for which we even require that any run of the ASM with the memory management behaves in the same way as a concurrent ASM, i.e. different from the view compatibility we even require that parallel read- and write requests can be seen as referring to the same state. Formally, this requires to strengthen the requirements for a flat view such that for parallel request the index $n$ (with $k < n \le m$) is always the same. Then the specification must be refined to handle also all parallel requests in parallel by the data cente agent.

\section{Refinement with Internal Communication}\label{sec:data}

We will now address a refinement of our specification of the memory management subsystem. So far in $\mathcal{CM}_1$ the \textit{home\/} data centre agent $d$ associated with an agent $a$ manages in one step the retrieval of data from or update of data of sufficiently many replicas as specified by the read- and write-policies, respectively. In doing so we abstracted from any internal communication between data centres.
 
However, in reality data centres refer to different physical machines, so the gist of the refinement is to treat the handling of a request as a combination of direct access to local nodes, remote access via messages to the other relevant data centres, and collection and processing of return messages until the requirements for the read- or write-policy are fulfilled. That is, the validation of the policy accompanies the preparation of a response message and is no longer under control of the \textit{home\/} agent.

We first show again that the refinement is a {\em complete} ASM refinement. Then we investigate again view compatibility and show that it is not preserved by the refinement. In fact, view compatibility is linked to the refinement being also {\em correct}, i.e. for each run of the concrete concurrent system we find a corresponding run of the abstract concurrent system. In general, however, this is not the case. We can even show that view compatibility implies view serialisability, which means that consistency (if desired) can only be guaranteed, if transactions (at least for single requests) are used. On the other hand, transactions together with appropriate read- and write-policies trivially imply view compatibility.

\subsection{Request Handling with Communicating Data Centres}
\label{sec:refinement}

In Section \ref{sec:ground} we specified how a data centre $d$
(treated as an agent) handles a request received from an agent $a$
with $\textit{home\/}(a) = d$. We distinguished read
requests read($p_i,\varphi$) subject to a read-policy and write
requests subject to a write-policy. Now, we first specify an abstract rule which manages
external requests, i.e. coming from an agent $a$ and received by any
data centre $d$, where \textit{request\/} is one of these read or
write requests. Essentially an external request is forwarded as
internal request to all other data centres $d^\prime \in \mathcal{D}_i$,
where it is handled and answered locally (see the definition of
\texttt{HandleLocally} below), whereas collecting (in $\text{answer}(p_i,\varphi)$) and
sending the overall answer to the external agent $a$ is delegated to a
new agent $a^\prime$. \texttt{SELF} denotes the data centre agent $d$
which executes the rule.

\begin{asm}
\texttt{DelegateExternalReq} = \\
\;\;\texttt{IF} 
        \texttt{RECEIVED}($\textit{request\/},\text{from:}a$) \texttt{THEN} \\
\;\;\;\;\texttt{LET} \> $t_{\text{current}} = \textit{clock}_{\texttt{SELF}}$ \texttt{IN}  \\
\;\;\;\;\texttt{LET} $a^\prime$ = \textit{new}(Agent)\/ \texttt{IN} \\
\> \texttt{Initialize}($a^\prime$) \\
\> \texttt{HandleLocally}($\textit{request\/},a^\prime,t_{\text{current}} $)\\
\> \texttt{ForwardToOthers}($\textit{request\/},a^\prime,t_{\text{current}} $)\\
\>  \texttt{CONSUME}($\textit{request\/},\text{from:}a$)\\
\texttt{WHERE}\\
\texttt{ForwardToOthers}($\textit{request\/},a^\prime,t_{\text{current}}$)=\\
\>\texttt{FORALL} $d^\prime \in \mathcal{D}_i$ \texttt{WITH} 
          $d^\prime \neq \texttt{SELF}$ \\
\>\> \texttt{SEND}$((\textit{request\/},a^\prime,t_{\text{current}}),\text{to:}d^\prime)$ \\
\texttt{Initialize}($a^\prime$) = \\
  \> $\text{answer}_{a^\prime}:= \emptyset$ \\
  \> \texttt{FORALL} $1 \le j \le q_i$ \\
\>     \texttt{FORALL} $d^\prime \in \mathcal{D}_i$ \\ 
   \>\> $\textit{count\/}_{a^\prime}(j,d^\prime) := 0 $\\
   \>\> $\textit{count\/}_{a^\prime}(j) := 0 $\\
   \>\> \texttt{IF} $request = \text{read}(p_i,\varphi)$ \\
\>\> \texttt{THEN} $\textit{asm}_{a^\prime}:=\texttt{CollectRespondToRead}$\\
   \>\> \texttt{ELSE} $\textit{asm}_{a^\prime}:=\texttt{CollectRespondToWrite}$\\
   \>\> $\textit{requestor}_{a^\prime}:=a$ \\
   \>\> $\textit{mediator}_{a^\prime}:=\texttt{SELF}$\\
   \>\> $\textit{requestType}_{a^\prime}:=\textit{request}$
\end{asm}

To \texttt{Initialize} a delegate $a$ it is equipped with a set $\text{answer}_a$, where to collect the values arriving from the asked data centres and with counters  $\textit{count\/}_a(j,d)$ (for the number of inspected replicas of the $j$-th fragment at data centre $d$) and $\textit{count\/}_a(j)$ (for the number of inspected replicas of the $j$-th fragment). The counters serve to check compliance with policies. The $mediator$ and $requestor$ information serves to retrieve sender and receiver once the delegate could complete the answer to the request it has been created for.

In this way the request handling agent $d$ simply forwards the request
to all other data centre agents and in parallel handles the request
locally for all nodes associated with $d$. The newly created agent (a
`delegate') will take care of collecting all response messages and
preparing the response to the issuing agent $a$. Request handling by
any other data centre $d^\prime$ is simply done locally using the
following rule.

\begin{asm}
\texttt{ManageInternalReq} = \\
\;\;\; \texttt{IF} 
   \texttt{RECEIVED}$((\textit{request\/},a^\prime,t),\text{from:}d)$  \texttt{THEN} \\
\>\texttt{HandleLocally}$(\textit{request\/},a^\prime,t)$\\
\>  \texttt{CONSUME}$((\textit{request\/},a^\prime,t),\text{from:}d)$
\end{asm}

That is we equip each data centre (agent) $d$ (of a cluster $\mathcal{D}_i$ related to $p_i$) with the following ASM program, where the components \texttt{HandleLocally} and the two versions of  \texttt{CollectRespond} are defined below.

\begin{asm}
$\as_d$ = \\
   \>\texttt{DelegateExternalReq} \\
   \> \texttt{ManageInternalReq}
\end{asm}

For local request handling we preserve most of what was
specified in Section \ref{sec:ground} with the difference that
checking the policy is not performed by the data centre agent but by
the delegate of the request; check the predicates
\textit{all\_messages\_received\/} and
\textit{sufficient\/}(policy) below. We use a predicate \textit{alive\/} to
check, whether a node is accessible or not. A possible interpretation
of this liveness concept is that a node is inaccessible if it does not
reply fast enough. For a read request we specify
\texttt{HandleLocally}$(\text{read}(p_i,\varphi),a^\prime,t_{\text{current}})$
as follows.

\begin{asm}
\texttt{HandleLocally}$(\text{read}(p_i,\varphi),a^\prime,t)$=\\
\;\;\; \texttt{LET}  $d^\prime=$ \texttt{SELF} \texttt{IN}\\
\;\;\; \texttt{FORALL} $j \in \{ 1,\dots,q_i \}$ \\
\> \texttt{LET} \> $G_{i,j,d^\prime} = 
    \{ j^\prime \mid \textit{copy\/}(i,j,d^\prime,j^\prime) \wedge 
    \textit{alive\/}(d^\prime,j^\prime) \}$ \texttt{IN}\\
\> \texttt{LET} \> $t_{\max}(\vec{k}) = 
      max(\{ t \mid \exists \vec{v}^\prime, \bar{j} . 
           \bar{j} \in G_{i,j,d^\prime} \wedge
            p_{i,j,d^\prime,\bar{j}}(\vec{k}) = (\vec{v}^\prime,t) \})$ \texttt{IN}\\
\> \texttt{LET} \> $\text{answer}_{i,j,d^\prime} = 
 \{ (\vec{k},\vec{v},t_{\max}(\vec{k})) \mid \varphi(\vec{k}) 
  \wedge h_i(\vec{k}) \in \textit{range}_j \;\wedge$ \\
\>\>\> $\exists j^\prime \in G_{i,j,d^\prime} .
      p_{i,j,d^\prime,j^\prime}(\vec{k}) = (\vec{v},t_{\max}(\vec{k})) \}$ \texttt{IN}\\
\> \texttt{LET} \> $\text{ans}({d^\prime}) = \bigcup_{j=1}^{q_i} \text{answer}_{i,j,d^\prime}$, 
$\vec{x} = (|G_{i,1,d^\prime}| , \dots, |G_{i,q_i,d^\prime}|)$  \texttt{IN}\\ 
\;\;\; \texttt{SEND}(answer$(\text{ans}({d^\prime}),\vec{x}),\text{to:}a^\prime$)
\end{asm}

Here we evaluate the request locally, but as the determined maximal
timestamp may not be globally maximal, it is part of the returned
relation. Also the number of replicas that are alive and thus
contributed to the local result is returned, such that the delegate
$a^\prime$ responsible for collection and final evaluation of the
request can check the satisfaction of the read-policy. For this the
created partial result is not returned to the agent $d$ that issued
this local request, but instead to the delegate (the collection agent).

For the write requests we proceed analogously. For the handling of an
update request the rule
\texttt{HandleLocally}$(\text{write}(p_i,p),a^\prime,t_{\text{current}})$ is
specified as follows.

\begin{asm}
\texttt{HandleLocally}$(\text{write}(p_i,p),a^\prime,t^\prime)$=\\
\;\;\; \texttt{LET}  $d^\prime= $\texttt{SELF} \texttt{IN}\\
\;\;\; \texttt{IF} $clock_{d^\prime} < t^\prime$  \texttt{THEN} adjust\_clock$(d^\prime,t^\prime)$ \\
\;\;\; \texttt{LET} \> $G_{i,d^\prime}(j) = 
   \{ j^\prime \mid \textit{copy\/}(i,j,d^\prime,j^\prime) \wedge 
   \textit{alive\/}(d^\prime,j^\prime) \}$ \texttt{IN}\\
\;\;\; \texttt{FORALL} $j \in \{ 1,\dots,q_i \}$ \\
\> \texttt{FORALL} $j^\prime \in G_{i,d^\prime}(j)$ \\
\> \texttt{FORALL}
  $(\vec{k},\vec{v}) \in p$ \texttt{WITH} $h_i(\vec{k}) \in \textit{range}_j$  \\
\>\;\;\; \texttt{IF} $\exists \vec{v}^\prime, t . p_{i,j,d^\prime,j^\prime}(\vec{k}) = (\vec{v}^\prime, t) \wedge t < t^\prime$ \\
\>\;\;\; \texttt{THEN} 
   $p_{i,j,d^\prime,j^\prime}(\vec{k}) := (\vec{v},t^\prime)$ \\
\;\;\; \texttt{LET} $\vec{x} = (|G_{i,1,d^\prime}| , \dots, |G_{i,q_i,d^\prime}|)$ \texttt{IN}\\ \>\texttt{SEND}$(\text{ack\_write}(p_i,\vec{x}),\text{to:}a^\prime$) \\
\end{asm}

Again, the partial results acknowledging the updates at the nodes
associated with data centre $d^\prime$ are sent to the collecting
agent $a^\prime$, which will verify the compliance with the
write-policy. Note that even a non-successful update at a location $(p_{i,j,d^\prime,j^\prime},\vec{k})$---this may result, if already a value with a larger timestamp exists---will be counted for $|G_{i,j,d^\prime}|$.

For the delegate $a^\prime$ that has been created by $d$ for a
request with the task to collect partial responses and to create the
final response to the agent $a$ issuing the request we need predicates \textit{sufficient\/}(policy) to check, whether the read and write policies are fulfilled, in which case the
response to $a$ is prepared and sent. These predicates are defined for
each $i$ as follows.
\begin{align*}
\textit{sufficient\/}(\textsc{All}) \; &\equiv \; \forall j . ( 1 \le j \le q_i \Rightarrow \textit{count\/}(j) = \gamma_{i,j} ) \\
& \qquad\qquad \text{with} \;\gamma_{i,j} = | \{ (d^\prime,j^\prime) \mid  \textit{copy\/}(i,j,d^\prime,j^\prime) \} | \\
\textit{sufficient\/}(\textsc{One}) \; &\equiv \; \forall j . ( 1 \le j \le q_i \Rightarrow \textit{count\/}(j) \ge 1 ) \\
\textit{sufficient\/}(\textsc{Two}) \; &\equiv \; \forall j . ( 1 \le j \le q_i \Rightarrow \textit{count\/}(j) \ge 2 ) \\
\textit{sufficient\/}(\textsc{Three}) \; &\equiv \; \forall j . ( 1 \le j \le q_i \Rightarrow \textit{count\/}(j) \ge 3 ) \\
\textit{sufficient\/}(\textsc{Quorum}(q)) \; &\equiv \; \forall j . ( 1 \le j \le q_i \Rightarrow \gamma_{i,j} \cdot q < \textit{count\/}(j) ) \\
\textit{sufficient\/}(\textsc{Each\_Quorum}(q)) \; &\equiv \; \forall j . ( 1 \le j \le q_i \Rightarrow \forall d \in \mathcal{D}_i . \delta_{i,j,d} \cdot q < \textit{count\/}(j,d) ) \\
& \qquad\qquad \text{with} \; \delta_{i,j,d} = | \{ j^\prime \mid  \textit{copy\/}(i,j,d,j^\prime) \} | \\
\textit{sufficient\/}(\textsc{Local\_Quorum}(q,d)) \; &\equiv \; \forall j . ( 1 \le j \le q_i \Rightarrow \delta_{i,j,d} \cdot q < \textit{count\/}(j,d) ) \\
\textit{sufficient\/}(\textsc{Local\_One}(d)) \; &\equiv \; \forall j . ( 1 \le j \le q_i \Rightarrow \textit{count\/}(j,d) \ge 1 )
\end{align*}

\noindent
It remains to specify the delegate rules $\texttt{CollectRespondToRead\/}$ and $\texttt{CollectRespondToWrite\/}$,
the programs associated with the agent $a^\prime$, which was created upon
receiving a $request$ from an agent $a$. The
$\texttt{CollectRespond\/}$ action is performed until all required messages have been received and splits into two rules for read and write requests, respectively.

While being alive, the delegate $a^\prime$ collects the messages
it receives from the data centres $d^\prime$ to which the original
$request$ had been forwarded to let them
\texttt{HandleLocally} $(request,a^\prime)$. If the set of collected
answers suffices to respond, the delegate sends an answer to the
original requester and kills itself.\footnote{We leave it to the
  garbage collector to deal with later arriving messages, that is
  messages addressed to a delegate which has been deleted already.}
Thus each of the rules of \texttt{CollectRespond} has a
\texttt{Collect} and a \texttt{Respond} subrule with corresponding
parameter for the type of expected messages.

\begin{asm}
\texttt{CollectRespondToRead}=\\
\;\;\; \texttt{IF} 
   \texttt{RECEIVED}$(answer(\text{ans}({d^\prime}),\vec{x}),\text{from:}d^\prime)$ \\
\> \texttt{THEN} \texttt{Collect}$((\text{ans}({d^\prime}),\vec{x}),\text{from:}d^\prime)$\\
\;\;\; \texttt{IF} $\textit{sufficient\/}(\text{read-policy})$ \\
\> \texttt{THEN} \texttt{Respond}$(requestType_{\texttt{SELF}})$\\
\texttt{WHERE}\\
\texttt{Respond}($\text{read}(p_i,\varphi)$)=\\
\> \texttt{LET} $d$=\textit{mediator}(\texttt{SELF}),\;
      $a$=\textit{requestor}(\texttt{SELF}) \texttt{IN}\\
\> \texttt{LET}   $\text{answer}(p_i,\varphi) = \{ (\vec{k},\vec{v}) \mid 
     \exists t .\; (\vec{k},\vec{v},t) \in \text{answer}_{\texttt{SELF}} \}$ \texttt{IN}\\
\>\> \texttt{SEND}($\text{answer}(p_i,\varphi),\text{from:}d,\text{to:}a)$\\
\> \texttt{DELETE}(\texttt{SELF},\texttt{Agent})\\
\texttt{Collect}$((\text{ans}({d^\prime}),\vec{x}),\text{from}:d^\prime)$=\\
\;\;\; \texttt{FORALL} $\vec{k}$ \texttt{WITH} 
   $\exists \vec{v}, t .\; (\vec{k},\vec{v},t) \in \text{ans}({d^\prime})$ \\
\> \texttt{LET} $(\vec{k},\vec{v},t) \in \text{ans}({d^\prime})$  \texttt{IN}\\
\>\> \texttt{IF}  $\exists \vec{v}^\prime , t^\prime .\; 
                (\vec{k},\vec{v}^\prime,t^\prime) \in \text{answer}(p_i,\varphi)$ \\
\>\> \texttt{THEN} \texttt{LET}  $(\vec{k},\vec{v}^\prime,t^\prime) \in \text{answer}(p_i,\varphi)$  \texttt{IN}\\
\>\>\>\texttt{IF} $t^\prime < t$ \\
\>\>\> \texttt{THEN} \> \texttt{DELETE}($(\vec{k},\vec{v}^\prime,t^\prime), \text{answer}(p_i,\varphi)$)\\
\>\>\>\>   \texttt{INSERT}$((\vec{k},\vec{v},t) , \text{answer}(p_i,\varphi))$\\
\>\>  \texttt{ELSE} \> \texttt{INSERT}$((\vec{k},\vec{v},t) , \text{answer}(p_i,\varphi))$\\
\;\;\; \texttt{LET} $(x_1,\dots,x_{q_i}) = \vec{x}$   \texttt{IN}\\
\> \texttt{FORALL} $j \in \{ 1,\dots,q_i \}$  \\
\>\> $\textit{count\/}(j,d^\prime) := \textit{count\/}(j,d^\prime) + x_j$ \\
\>\>$\textit{count\/}(j) := \textit{count\/}(j) + x_j$ \\
\;\;\;\texttt{CONSUME}$((\text{ans}({d^\prime}),\vec{x}),\text{from:}d^\prime)$
\end{asm}

The analogous collection of messages for write
requests is simpler, as the final response is only an
acknowledgement. 

\begin{asm}
\texttt{CollectRespondToWrite}=\\
\;\;\; \texttt{IF} 
    \texttt{RECEIVED}$(\text{ack\_write}(p_i,\vec{x}),\text{from:}d^\prime)$ \texttt{THEN}  \\
\> \texttt{Collect}$(\text{ack\_write}(p_i,\vec{x}),\text{from:}d^\prime)$ \\
\> \texttt{CONSUME}$(\text{ack\_write}(p_i,\vec{x}),\text{from:}d^\prime)$\\
\;\;\; \texttt{IF} \> $\textit{sufficient\/}(\text{write-policy})$  \texttt{THEN}\\
 \> 
     \texttt{SEND}$(\text{acknowledge}
 (\text{write},p_i),\text{from:}\textit{mediator}(\texttt{SELF}),
       \text{to:}\textit{requestor}(\texttt{SELF}))$ \\
 \> \texttt{DELETE}(\texttt{SELF},\texttt{Agent})\\
\texttt{WHERE}\\
\texttt{Collect}$(\text{ack\_write}(p_i,\vec{x}),\text{from:}d^\prime)$ =\\
\>\texttt{LET} $(x_1,\dots,x_{q_i}) = \vec{x}$  \texttt{IN}\\
\> \>\texttt{FORALL} $j \in \{ 1,\dots,q_i \}$  \\
 \> \>\>$\textit{count\/}(j,d^\prime) := \textit{count\/}(j,d^\prime) + x_j$ \\
 \>\> \>$\textit{count\/}(j) := \textit{count\/}(j) + x_j$ 
\end{asm}

Note that in our specification we do not yet deal with exception handling. We may tacitly assume that eventually all requested answers will be received by the collecting agent.

\subsection{Analysis of the Refinement}

First we investigate the interaction of one agent $a$ with the database. By $\mathcal{CM}_1$ we already denoted the abstract concurrent ASM $\{ (a,\as^c_a) \}_{a \in \mathcal{A}} \cup 
\{ (d,\as_d) \}_{d \in \mathcal{D}}$ from Section \ref{sec:ground}. Now
let $\mathcal{CM}_2$ denote the refined concurrent ASM $\{
(a,\as^c_a) \}_{a \in \mathcal{A}} \cup \{ (d,\as_d^\prime) \}_{d \in
\mathcal{D} \cup \mathcal{E}xt}$ together with the dynamic set $ \mathcal{E}xt$ of delegates from Subsection
\ref{sec:refinement}. Note that the delegates are created on-the-fly by the agents
$d \in \mathcal{D}$, needed for collecting partial responses for
each request and preparing the final responses.

Let $S_0, S_1, \dots$ be a concurrent run of
$\mathcal{CM}_1$, which we first look at from the
perspective of a single agent $a \in \mathcal{A}$. In each of its non purely local
steps $\as^c_a$ produces an update set, which together with update sets
produced by other agents is applied to some state (one of the
$S_j$). Let $S$ (another $S_i$ with $j < i$) denote the resulting
state. Apart from updates to non-shared locations with function
symbols in $\Sigma_a$ the updates brought into the state $S$ are read
and write requests sent to the memory management
subsystem. Then agent $a$ continues in some state $S^\prime = S_{i+x}$
evaluating the responses to the read and write requests it had sent. Furthermore,
there is another state $S^{\prime\prime} = S_{i+y}$ ($y \ge x$), where
all updates issued by write requests of $a$ have been propagated to
all replicas.

The transition from $S$ to $S^\prime$ is achieved by means of a step
of $\as_d$ for $d = \text{\textit{home\/}}(a)$, and the same holds for the
transition from $S^\prime$ to $S^{\prime\prime}$. Therefore, there
exist corresponding states $\bar{S}$, $\bar{S}^\prime$ for
$\mathcal{CM}_2$, in which the unchanged $\as^c_a$ brings in
the same read and write requests and receives the last response,
respectively. In a concurrent run for $\mathcal{CM}_2$ the
transition from $\bar{S}$ to $\bar{S}^\prime$ results from several
steps by the subsystem $\{ (d,\as_d^\prime) \}_{d \in \mathcal{D} \cup
  \mathcal{E}xt}$.

With respect to each of the requests received from $a$ the agent $d =
\text{\textit{home\/}}(a)$ contributes to a state $\bar{S}_1$ with requests for
each agent $d^\prime \in \mathcal{D}$, $d^\prime \neq d$, the creation
and initialisation of a response collecting agent $a^\prime$, and the
local handling of the request at nodes associated with data centre
$d$. Then each agent $d^\prime \in D$ contributes to some state
$\bar{S}_k$ ($k > 1$), in which the partial response to the request
sent to $d^\prime$ is produced. Concurrently the collection agent
$a^\prime$ on receipt of a partial response updates its own
locations---these comprise the read response $\varrho$, the acknowledgement responses and several
counters---and thus contributes to some state $\bar{S}_j$ ($j >
1$). Finally, $a^\prime$ will also produce and send the response to
$a$. This response will be the same as the one in state $S^\prime$, if
the refined run from $\bar{S}$ to $\bar{S}^\prime$ uses the same
selection of copies for each request referring to $p_i$ and each
fragment $\textit{Frag}_{j,i}$.

Furthermore, some partial responses may arrive in states $\bar{S}_j$
with $j>k$ for $\bar{S}^\prime = \bar{S}_k$. These are not processed by $a^\prime$ any more, but by the garbage collector. Thus the state $\bar{S}^{\prime\prime}$ , when all updates have been
propagated, corresponds to state
$S^{\prime\prime}$. Taking these considerations together we have shown
the following lemma.

\begin{lemma}

The concurrent ASM $\mathcal{CM}_{2,a} =  \{ (a,\as^c_a) \}\cup
\{ (d,\as_d^\prime) \}_{d \in \mathcal{D} \cup \mathcal{E}xt}$ is a
complete refinement of $\mathcal{CM}_{1,a} = \{ (a,\as^c_a) \}
\cup \{ (d,\as_d) \}_{d \in \mathcal{D}}$.

\end{lemma}

Now consider the actions of all agents $a \in \mathcal{A}$ together. The
arguments above remain valid, except that for the sequence $\bar{S},
\bar{S}_1, \bar{S}_2, \dots, \bar{S}^\prime, \dots ,
\bar{S}^{\prime\prime}$ there are many intermediate states, in which
updates of other agents and other data centres for other
requests are brought in. A partial response for a request by $a$ and
the effects on values in the replicas may in general differ, as
different partial responses may result depending on the order of
request handling. However, there exists a run of
$\mathcal{CM}_2$ such that the projection to
$\mathcal{CM}_{2,a}$ is exactly the sequence $\bar{S},
\bar{S}_1, \bar{S}_2, \dots, \bar{S}^\prime, \dots ,
\bar{S}^{\prime\prime}$ investigated before. This gives us the
following result.

\begin{proposition}\label{prop:3}

$\mathcal{CM}_2 =  \{ (a,\as^c_a) \}_{a \in \mathcal{A}} \cup
 \{ (d,\as_d^\prime) \}_{d \in \mathcal{D} \cup \mathcal{E}xt}$ is a
  complete refinement of $\mathcal{CM}_1 =  \{ (a,\as^c_a) \}_{a \in \mathcal{A}} \cup \{ (d,\as_d) \}_{d \in \mathcal{D}}$.

\end{proposition}

\subsection{Consistency Analysis}

While we have obtained a complete refinement, $\mathcal{CM}_2$ permits more possibilities to perform updates to replicas and reading values from replicas, respectively, in different orders. Unfortunately, view compatibility (as in Proposition \ref{prop:1}) cannot be preserved, as the following simple counterexample shows:

Consider a single location $(p_i,\vec{k}))$ for a fixed $\vec{k}$. For simplicity forget about the relation symbols $p_i$ and let this location simply be $x$. Assume that there are at least two replicas (say $x_1$ and $x_2$), which both are set to the initial value $0$ (with some timestamp $t_0$). Let the write policy be \textsc{All} and the read policy be \textsc{One} (which according to Proposition \ref{prop:1} is an appropriate combination of policies that leads to view compatibility for the abstract specification $\mathcal{CM}_1$). 

Let $a_1$ be an agent issuing a write request \texttt{SEND}(write($x$,1)), and let $a_2$ be another agent issuing a sequence of two read requests \texttt{SEND}(read($x$)). Then in $\mathcal{CM}_2$ the following sequence is possible:

\begin{enumerate}

\item The write request by $a_1$ gives rise to the update of replica $x_1$, so the value is set to $(1,t)$ with some timestamp $t > t_0$.

\item For the first read request by $a_2$ only the replica $x_1$ is evaluated, and thus the returned answer is $1$.

\item For the following second read request by $a_2$ only the replica $x_2$ is evaluated, and thus the returned answer is $0$.

\item All other updates required for the completion of the write request by $a_1$ happen after the completion of the two read requests.

\end{enumerate}

Obviously, there cannot be a run of the communicating concurrent ASM without replication that produces the same answers to the read and write requests.

From the counterexample we can even conclude that there is no reasonable weaker definition of consistency such that an appropriate combination of read- and write-policies alone, even under strong assumptions that no bulk requests are considered, suffices to guarantee consistency. For such a definition the run in the counterexample would have to be called ``consistent'', which does not make much sense.

The question is whether additional conditions could be enforced in the specification that would lead again to view compatibility. As we will see, any such condition already implies serialisability. Therefore, we conclude our analysis by showing this relationship between view compatibility and (view) serialisability, which implies that the best and well explored way to ensure consistency is to exploit transactions (at least for single read and write requests).

Let us first define the notion of serialisability. For this consider an arbitrary run $\mathcal{R} = S_0, S_1, \dots$ of $\mathcal{CM}_2$. A {\em request} is either a read request of the form $(\text{read}(p_i,\varphi),\text{from:}a)$ or a write request of the form $(\text{write}(p_i,p),\text{from:}a)$. If $S_j$ is the first state, in which \texttt{RECEIVED}($r$) holds for such a request $r$, we call $S_j$ the state, in which the request $r$ is issued, and write $\sigma(r) = j$ as well as $ag(r) = a$. Analogously, a {\em response} $r$ takes either the form $(\text{answer}(p_i,\varphi),\text{to:}a)$ or $(\text{acknowledge}(\text{write},p_i),\text{to:}a)$. We call the state $S_k$, in which \texttt{SEND}($r$) becomes effective, the state, in which the response $r$ is issued, and write also $\sigma(r) = j$ as well as $ag(r) = a$. Then $\sigma$ defines a partial order $\preceq$ on the set of requests and responses in the run as well as an equivalence relation $\sim$ for {\em simultaneous} requests and responses. Furthermore, there is a bijection \textit{ans\/} that maps each request to its corresponding response.

In a serial run a request is immediately followed by its corresponding response without any other request or response in between, but requests may be handled in parallel. Thus, formally we call the run $\mathcal{R}$ {\em serial} iff the following two conditions hold:

\begin{enumerate}

\item If $r$ is a request and $r^\prime$ is another request or response with $\sigma(r) \le \sigma(r^\prime) \le \sigma(\textit{ans\/}(r))$, then either $r \sim r^\prime$ or $r^\prime \sim \textit{ans\/}(r)$ hold.

\item For two simultaneous requests $r \sim r^\prime$ we also have $\textit{ans\/}(r) \sim \textit{ans\/}(r^\prime)$.

\end{enumerate}

\begin{definition}\rm

Two runs $\mathcal{R}$ and $\mathcal{R}^\prime$ are called {\em view equivalent} iff the following two conditions hold:

\begin{enumerate}

\item $\mathcal{R}$ and $\mathcal{R}^\prime$ contain exactly the same requests and responses. In particular, in a response $(\text{answer}(p_i,\varphi),\text{to:}a)$ to a read request $(\text{read}(p_i,\varphi),\text{from:}a)$ the sets $\text{answer}(p_i,\varphi)$ of tuples are identical in both runs.

\item For each agent $a$ the sequence of its requests and responses is identical in both runs, i.e. if $\sigma(r) < \sigma(r^\prime) \wedge ag(r) = ag(r^\prime) = a$ holds in $\mathcal{R}$, it also holds in $\mathcal{R}^\prime$.

\end{enumerate}

Then a run $\mathcal{R}$ is called {\em view serialisable} iff there exists a serial run $\mathcal{R}^\prime$ that is view equivalent to $\mathcal{R}$. 

\end{definition}

Informally phrased, a run is view serialisable if there exists a serial run, in which for each agent the requests and corresponding responses appear in the same order and the same results are produced.

\begin{proposition}\label{prop:serialisability1}

\ If $\mathcal{CM}_2$ is view compatible with the concurrent ASM $\mathcal{CM}_0 = \{ (a,\as^c_a) \}_{a \in \mathcal{A}} \cup \{ (db,\as_{db}) \}$, then every run $\mathcal{R}$ of $\mathcal{CM}_2$ is view serialisable.

\end{proposition}

\begin{proof}

Let $\mathcal{R} = S_0, S_1, \dots$ be a run of $\mathcal{CM}_2$. The definition of view compatibility implies that there exists a subsequence of a flattening $\mathcal{R}^\prime = S_0^\prime, S_1^\prime, S_2^\prime, \dots$ that is a run of $\{ (a,\as^c_a) \}_{a \in \mathcal{A}} \cup \{ (db,\as_{db}) \}$ such that for each agent $a \in \mathcal{A}$ the agent $a$-view of $\mathcal{R}^\prime$ coincides with a flat view of $\mathcal{R}$ by $a$. Let $\Delta_\ell^\prime$ be the update set defined by $S_\ell^\prime + \Delta_\ell^\prime = S_{\ell+1}^\prime$. Define
\[ \Delta_\ell = \{ ((p_{i,j,d,j^\prime},\vec{k}),(\vec{v},t_\ell)) \mid ((p_i,\vec{k}),\vec{v}) \in \Delta_\ell^\prime \wedge \textit{copy\/}(i,j,d,j^\prime) \wedge h_i(\vec{k}) \in \textit{range\/}_j \} \]

using timestamps $t_0 < t_1 < t_2 < \dots$. This defines a run $\bar{\mathcal{R}} = \bar{S}_0, \bar{S}_1, \dots$ of $\mathcal{CM}_2$ with $\bar{S}_0 = S_0$ and $\bar{S}_{\ell + 1} = \bar{S}_\ell + \Delta_\ell$. In this run the answer to a read request is executed in a single step and all updates requested by a write request are executed in a single step. Thus $\mathcal{R}^\prime$ is serial.

Furthermore, as for each agent $a \in \mathcal{A}$ the agent $a$-view of $\mathcal{R}^\prime$ coincides with a flat view of $\mathcal{R}$ by $a$, the runs runs $\mathcal{R}$ and $\bar{\mathcal{R}}$ contain exactly the same requests and responses, and for each agent $a$ the sequence of its requests and responses is identical in both runs. That is, $\mathcal{R}$ and $\bar{\mathcal{R}}$ are view equivalent.\qed

\end{proof}

We can finally also show the inverse of the implication in Proposition \ref{prop:serialisability1}, provided that an appropriate combination of read- and write-policies is used.

\begin{proposition}\label{prop:serialisability2}

If all runs of $\mathcal{CM}_2$ are view serialisable and an appropriate combination of a read and a write policy is used, then $\mathcal{CM}_2$ is also view compatible with the concurrent ASM $\mathcal{CM}_0 = \{ (a,\as^c_a) \}_{a \in \mathcal{A}} \cup \{ (db,\as_{db}) \}$.

\end{proposition}

\begin{proof}

\ Let $\mathcal{R}^\prime$ be a run of $\mathcal{CM}_2$ and let $\mathcal{R}^\prime = S_0, S_1, \dots$ be a view equivalent serial run. It suffices to show that that there exists a subsequence of a flattening $\mathcal{R}^\prime = S_0^\prime, S_1^\prime, S_2^\prime, \dots$ that is a run of $\{ (a,\as^c_a) \}_{a \in \mathcal{A}} \cup \{ (db,\as_{db}) \}$ such that for each agent $a \in \mathcal{A}$ the agent $a$-view of $\mathcal{R}^\prime$ coincides with a flat view of $\mathcal{R}$ by $a$.

For this we only have to consider states, in which a request or a response is issued to obtain the desired subsequence, and the flattening is defined by the answers to the write requests. As $\mathcal{R}$ is serial, the condition that for each agent $a$ the $a$-view concides with a flat $a$-view follows immediately.\qed

\end{proof}

\section{Conclusions}\label{sec:schluss}

In this paper we demonstrated the maturity of concurrent communicating ASMs (ccASMs) \cite{boerger:ai2016,boerger:jucs2017} showing how they can be used to specify and analyse concurrent computing systems in connection with shared replicated memory. Using ccASMs we first specified a ground model, in which all access to replicas is handled synchronously in parallel by a single agent.

We then refined our ground model addressing the internal communication in
the memory management subsystem. This refinement significantly changes
the way requests are handled, as the replicas are not selected a
priori in a way that complies with the read- or write-policy, but
instead the acknowledgement and return of a response depends on these
policies. This adds an additional level of flexibility to the internal
request handling. 

We used the specification to analyse consistency and showed that consistency, formalised by the notion of view compatibility, cannot be preserved by the refinement. To the contrary, we could show that even such a rather weak notion of consistency can only be obtained by adopting transactions for at least single requests. 

The refinements could be taken further to capture more and more details of the physical data organisation. For instance, we did not yet tackle the means for handling inactive nodes and for recovery. Nonetheless, while our refined specification is still rather abstract, it shows the way how concurrent systems interact with replicative memory management subsystems, and permits analysis of
which consistency level can be obtained.

It seems straightforward to combine a transactional concurrent system as for example in
\cite{boerger:scp2016} with the specification of a
replicative storage system as done in this paper. In particular, it is
well known that replication strategies can be easily combined with
transaction management. This would then enable stronger consistency results. 

\bibliographystyle{abbrv}
\bibliography{casm}

\begin{thebibliography}{10}

\bibitem{agha:1986}
G.~Agha.
\newblock {\em A Model of Concurrent Computation in Distributed Systems}.
\newblock MIT Press, Cambridge, Mass., 1986.

\bibitem{an:mth2016}
W.~An.
\newblock Formal specification and analysis of asynchronous mutual exclusion
  algorithms.
\newblock Master's thesis, JKU Linz, Austria, 2016.

\bibitem{best:1996}
E.~Best.
\newblock {\em Semantics of sequential and parallel programs}.
\newblock Prentice Hall, 1996.

\bibitem{blass:tocl2003}
A.~Blass and Y.~Gurevich.
\newblock {Abstract State Machines} capture parallel algorithms.
\newblock {\em ACM Trans. Computational Logic}, 4(4):578--651, 2003.

\bibitem{blass:tocl2008}
A.~Blass and Y.~Gurevich.
\newblock Abstract {S}tate {M}achines capture parallel algorithms: Correction
  and extension.
\newblock {\em ACM Trans. Comp. Logic}, 9(3), 2008.

\bibitem{boerger:jcss2012}
E.~B{\"{o}}rger, A.~Cisternino, and V.~Gervasi.
\newblock Ambient abstract state machines with applications.
\newblock {\em J. Comput. Syst. Sci.}, 78(3):939--959, 2012.

\bibitem{boerger:ai2016}
E.~B{\"{o}}rger and K.-D. Schewe.
\newblock Concurrent abstract state machines.
\newblock {\em Acta Inf.}, 53(5):469--492, 2016.

\bibitem{boerger:jucs2017}
E.~B{\"{o}}rger and K.-D. Schewe.
\newblock Communication in {Abstract State Machines}.
\newblock {\em J. Univ. Comp. Sci.}, 23(2):129--145, 2017.

\bibitem{boerger:scp2016}
E.~B{\"{o}}rger, K.-D. Schewe, and Q.~Wang.
\newblock Serialisable multi-level transaction control: {A} specification and
  verification.
\newblock {\em Sci. Comput. Program.}, 131:42--58, 2016.

\bibitem{boerger:2003}
E.~B{\"o}rger and R.~F. St{\"a}rk.
\newblock {\em Abstract {S}tate {M}achines. A Method for High-Level System
  Design and Analysis}.
\newblock Springer, 2003.

\bibitem{ferrarotti:tcs2016}
F.~Ferrarotti, K.-D. Schewe, L.~Tec, and Q.~Wang.
\newblock A new thesis concerning synchronised parallel computing -- simplified
  parallel {ASM} thesis.
\newblock {\em Theor. Comp. Sci.}, 649:25--53, 2016.

\bibitem{genrich:tcs1981}
H.~J. Genrich and K.~Lautenbach.
\newblock System modelling with high-level {P}etri nets.
\newblock {\em Theoretical Computer Science}, 13:109--136, 1981.

\bibitem{gurevich:lipari1995}
Y.~Gurevich.
\newblock Evolving algebras 1993: Lipari guide.
\newblock In E.~B\"{o}rger, editor, {\em Specification and Validation Methods}.
  Oxford University Press, 1995.

\bibitem{gurevich:tocl2000}
Y.~Gurevich.
\newblock Sequential abstract-state machines capture sequential algorithms.
\newblock {\em ACM Trans. Comp. Logic}, 1(1):77--111, 2000.

\bibitem{lamport:cacm1978}
L.~Lamport.
\newblock Time, clocks, and the ordering of events in a distributed system.
\newblock {\em Commun. ACM}, 21(7):558--565, 1978.

\bibitem{lamport:1990}
L.~Lamport and N.~Lynch.
\newblock {\em Distributed Computing: Models and Methods}, chapter Handbook of
  Theoretical Computer Science, pages 1157--1199.
\newblock Elsevier, 1990.

\bibitem{lynch:1996}
N.~Lynch.
\newblock {\em Distributed Algorithms}.
\newblock Morgan Kaufmann, 1996.

\bibitem{mazurkiewicz:lncs1987}
A.~Mazurkiewicz.
\newblock Trace theory.
\newblock volume 255 of {\em LNCS}, pages 279--324. Springer, 1987.

\bibitem{oezsu:2011}
M.~T. \"{O}zsu and P.~Valduriez.
\newblock {\em Principles of Distributed Database Systems, Third Edition}.
\newblock Springer, 2011.

\bibitem{prinz:abz2014}
A.~Prinz and E.~Sherratt.
\newblock Distributed {ASM} -- pitfalls and solutions.
\newblock In Y.~A\"{\i}t-Ameur and K.-D. Schewe, editors, {\em Abstract State
  Machines, Alloy, B, TLA, VDM and Z -- Proceedings of the 4th International
  Conference (ABZ 2014)}, volume 8477 of {\em LNCS}, pages 210--215. Springer,
  2014.

\bibitem{rabl:pvldb2012}
T.~Rabl, M.~Sadoghi, H.-A. Jacobsen, S.~{G{\'{o}}mez-Villamor},
  V.~{Munt{\'{e}}s-Mulero}, and S.~Mankowskii.
\newblock Solving big data challenges for enterprise application performance
  management.
\newblock {\em {PVLDB}}, 5(12):1724--1735, 2012.

\bibitem{schewe:acsw2017}
K.-D. Schewe, F.~Ferrarotti, L.~Tec, Q.~Wang, and W.~An.
\newblock Evolving concurrent systems -- behavioural theory and logic.
\newblock In {\em Proceedings of the Australasian Computer Science Week
  Multiconference ({ACSW} 2017)}, pages 77:1--77:10. ACM, 2017.

\bibitem{schewe:ac2010}
K.-D. Schewe and Q.~Wang.
\newblock A customised {ASM} thesis for database transformations.
\newblock {\em Acta Cyb.}, 19(4):765--805, 2010.

\bibitem{schewe:jucs2010}
K.-D. Schewe and Q.~Wang.
\newblock {XML} database transformations.
\newblock {\em J. Univ. Comp. Sci.}, 16(20):3043--3072, 2010.

\bibitem{tanenbaum:2007}
A.~S. Tanenbaum and M.~{Van Steen}.
\newblock {\em Distributed systems}.
\newblock Prentice-Hall, 2007.

\bibitem{winskel:1995}
G.~Winskel and M.~Nielsen.
\newblock Models for concurrency.
\newblock In S.~Abramsky, D.~Gabbay, and T.~S.~E. Maibaum, editors, {\em
  Handbook of Logic and the Foundations of Computer Science: Semantic
  Modelling}, volume~4, pages 1--148. Oxford University Press, 1995.

\end{thebibliography}

\end{document}